\documentclass[runningheads,a4paper]{llncs}

\usepackage{amssymb}
\usepackage{graphicx}

\usepackage[utf8]{inputenc}
\usepackage{amsfonts,amsmath} 
\usepackage[hidelinks]{hyperref}
\usepackage{enumerate}
\usepackage{algorithm2e}
\RestyleAlgo{boxruled}
\usepackage{caption}
\usepackage{subcaption}

\usepackage{url}

\newtheorem{thm}{Theorem}
\newtheorem{lem}[thm]{Lemma}
\newtheorem{cor}[thm]{Corollary}

\begin{document}

\mainmatter  

\title{Dijkstra Graphs}

\titlerunning{Dijkstra Graphs}

%
%
\author{Lucila M. S. Bento
$^1$
\and Davidson R. Boccardo$^{6,7}$\and Raphael C. S. Machado$^3$\and\\ Fl\'avio K. Miyazawa$^5$\and Vin\'icius G. Pereira de S\'a$^2$\and
Jayme L. Szwarcfiter$^{1,2,4}$}
\authorrunning{Dijkstra Graphs}

\institute{$^1$COPPE-Sistemas -- Universidade Federal do Rio de Janeiro \\
$^2$Instituto de Matem\'atica -- Universidade Federal do Rio de Janeiro \\
$^3$Instituto Nacional de Metrologia, Qualidade e Tecnologia -- Inmetro \\
$^4$Instituto de Matem\'atica e Estat\'istica -- Universidade do Estado do Rio de Janeiro \\
$^5$Instituto de Computa\c c\~ao -- Universidade Estadual de Campinas \\
$^6$Clavis Seguran\c ca da Informa\c c\~ao \\
$^7$Green Hat Seguran\c ca da Informa\c c\~ao\\
E-mails: lucilabento@ppgi.ufrj.br, davidson@clavis.com.br, rcmachado@inmetro.gov.br, fkm@ic.unicamp.br, vigusmao@dcc.ufrj.br, jayme@nce.ufrj.br}

%
%

\maketitle

\begin{abstract}
We revisit a concept that has been central in some early stages of computer science, that of {\it structured programming}: a set of rules that an algorithm must follow in order to acquire a structure that is desirable in many aspects. While much has been written about structured programming, an important issue has been left unanswered: given an arbitrary, compiled program, describe an algorithm to decide whether or not it is \emph{structured}, that is, whether it conforms to the stated principles of structured programming. We refer to the classical concept of structured programming, as described by Dijkstra. By employing a graph model and graph-theoretic techniques, we formulate an efficient algorithm for answering this question. To do so, we first introduce the class of graphs which correspond to structured programs, which we call {\it Dijkstra Graphs}. Our problem then becomes the recognition of such graphs, for which we present a greedy $O(n)$-time algorithm. Furthermore, we describe an isomorphism algorithm for Dijkstra graphs, whose complexity is also linear in the number of vertices of the graph. Both the recognition and isomorphism algorithms have potential important applications, such as in code similarity analysis.  
\end{abstract}

\section{Introduction}
Structured programming was one of the main topics in computer science in the years around 1970. 
It can be viewed as a method for the development and description of algorithms and programs.    
Basically, it consists of a top-down formulation of the algorithm, breaking it into blocks or modules. 
The blocks are stepwise refined, possibly generating new, smaller blocks, until refinements no longer exist. 
The technique constraints the description of the modules to contain only three basic control structures: \emph{sequence}, \emph{selection} and \emph{iteration}. 
The first of them corresponds to sequential statements of the algorithm; the second refers to comparisons leading to different outcomes; the last one corresponds to sets of actions performed repeatedly in the algorithm.  

One of the early papers about structured programming was the article by Dijkstra ``Go-to statement
considered harmful''~\cite{D68}, which brought the idea that the unrestricted use of go-to statements is incompatible with well structured algorithms. 
That paper was soon followed by a discussion in the literature about go-to's, as in the papers by Knuth and Floyd~\cite{KF71}, Wulf~\cite{Wu72} and Knuth~\cite{K74}. 
Other classical papers are those by Dahl and Hoare~\cite{DH72}, Hoare~\cite{Ho72} and Wirth~\cite{Wi71}, among others. 
Guidelines of structured programming were established in an article by Dijkstra~\cite{D72}.  The early development of programming languages containing blocks, such as ALGOL (Wirth~\cite{Wi71a}) and PASCAL (Naur~\cite{Na60}), was an important reason for structured programming's widespread adoption. This concept has been then further developed in papers by Kosaroju~\cite{Ko74},  describing the idea of reducibility among flowcharts. Moreover, ~\cite{Ko74}  has introduced and characterized the class of \emph{D-charts}, which in fact are graphs properly containing all those which originate from structured programming. Williams~\cite{Wi83} also describes variations of different forms of structuredness, including the basic definitions by Dijkstra, as well as D-charts.    
 The different forms of unstructuredness were described in papers by Williams~\cite{Wi77} and McCabe~\cite{Mc76}. The conversion of a unstructured flow diagram into a structured one has been considered by Williams and Ossher~\cite{WiOs78}, and Oulsnam~\cite{Ou82}. Formal aspects of structured programming include the papers by B\"ohm and Jacopini~\cite{BoJa66}, Harel~\cite{Ha80}, and Kozen and Tseng~\cite{KoTs08}. A mathematical theory for modeling structuredness, designed for flow graphs, in general, has been described by Fenton, Whitty and Kaposi~\cite{FWK85}.   
The actual influence of the concept of structured programming in the development of algorithms for solving various problems in different areas occurred right from the start, either explicitly, as in the papers by Henderson and Snow~\cite{HS72}, and Knuth and Szwarcfiter~\cite{KS74}, or implicitly as in the various graph algorithms by Tarjan, e.g.~\cite{Ta72,Ta74}.   
 
A natural question regarding structured programming is to recognize whether a given program is structured. 
To our knowledge, such a question has not been solved neither in the early stages of structured programming, nor later. 
That is the main purpose of the present paper. 
We formulate an algorithm for recognizing whether a given program is structured, according to Dijkstra's concept of structured programming. Note that the input comprises the binary code, not the source code. 
A well-known representation that comes in handy is that of the \emph{control  graph} (CFG) of a program, employed by the majority of reverse-engineering tools to perform data-flow analysis and optimizations. A CFG represents the intraprocedural computation of a function by depicting the existing links across its basic blocks. Each basic block represents a straight line in the program's instructions, ending (possibly) with a branch. 
An edge $A \rightarrow B$ (from the exit of block A to the start of block B) represents the program flowing from A to B at runtime. 

We are then interested in the version of the recognition problem which takes as input a control flow graph of the program~\cite{AC76,CD78}: a directed graph representing the possible sequences of basic blocks along the execution of the program. Our problem thus becomes  graph-theoretic: given a control flow graph, decide whether it has been produced by a structured program. We apply a reducibility method, whose reduction operations iteratively obtain smaller and smaller control flow graphs.    

%


In this paper, we first define the class of graphs which correspond to structured programs, as considered by Dijkstra in~\cite{D72}. Such a class has then been named as {\it Dijkstra graphs}.
We describe a characterization that leads to a greedy $O(n)$ time recognition algorithm for a Dijkstra graph with $n$ vertices.
Among the potential direct applications of the proposed algorithm, we can mention software watermarking via control flow graph modifications~\cite{BBMPS13,CKCT03}.

Additionally, we formulate an isomorphism algorithm for the class of Dijkstra graphs. The method consists of defining a convenient code for a graph of the class, which consists of a string of integers. Such a code uniquely identifies the graph, and it is shown that two Dijkstra graphs are isomorphic if and only if their codes are the same. The code itself has size $O(n)$ and the time complexity of the isomorphism algorithm is also $O(n)$. In case the given graphs are isomorphic, the algorithm exhibits the isomorphism function between the graphs. Applications of isomorphism include code similarity analysis~\cite{Co10}, since the method can determine whether apparently distinct control flow graphs (of structured programs) are actually structurally identical, with potential implications in digital rights management.




\section{Preliminaries}\label{preliminaries}

In this paper, all graphs are finite and directed.
For a graph $G$, we denote its vertex and edge sets by $V(G)$ and $E(G)$, respectively, with $|V(G)| = n$, $|E(G)| = m$. For $v,w \in V(G)$, an edge from $v$ to $w$ is written as $vw$.
We say $vw$ is an {\it out-edge} of $v$ and an {\it in-edge} of $w$, with $w$ an {\it out-neighbor} of $v$, and $v$ an {\it in-neighbor} of $w$. 
We denote by  $N^+_G(v)$ and $N^-_G(v)$ the sets of out-neighbors and in-neighbors of $v$, respectively. We may drop the subscript when the graph is clear from the context. Also, we write $N^{2+}(v)$ meaning $N^+(N^+(v))$.  For $v,w \in V(G)$, $v$ {\it reaches} $w$ when there is a path in $G$ from $v$ to $w$. A  {\it source} of $G$ is a vertex
that reaches all other vertices in $G$, while a {\it sink} is one which reaches no vertex, except itself. 
 Denote by $s(G)$ and $t(G)$, respectively, a source and a sink of $G$.   A {\it (control) flow} graph $G$ is one which contains a distinguished source $s(G)$. A {\it source-sink} graph contains both a distinguished source $s(G)$ and distinguished sink $t(G)$. A {\it trivial} graph contains a single vertex. 

A graph with no directed cycles is called {\it acyclic}. In an acyclic graph if there is a path from vertex $v$ to vertex $w$, then $v$ is an ancestor of $w$, and the latter a descendant of $v$. Additionally, if $v,w$ are distinct then $v$ is a {it proper ancestor}, and $w$ a {\it proper descendant}. 
Let $G$ be a flow graph with source $s(G)$, and $C$ a cycle of $G$. The cycle $C$ is called a {\it single-entry cycle} if it contains a vertex $v \in C$ that separates $s(G)$ from the vertices of $C \setminus \{v\}.$ A flow graph in which each of its cycles is a single-entry cycle is called {\it reducible}. Reducible graphs were characterized by Hecht and Ullman~\cite{HU72,HU74}. An efficient recognition algorithm for this class has been described by Tarjan~\cite{Ta74a}. 

In a {\it depth-first search (DFS)} of a directed graph, in each step a vertex is inserted in a stack, or  removed from it. Every vertex is inserted and removed from the stack exactly once. An edge $vw \in E(G)$, such that $v$ is inserted in the stack after $w$, and before the removal of $w$, is called a {\it cycle edge}. Let $C$ be the set of cycle edges of a graph, relative to some DFS. Clearly, the graph $G - C$ is acyclic. The following characterization if reducible flow graphs is relevant for our purposes.  

\begin{thm}\cite{HU74,Ta74a} A flow graph $G$ is reducible if and only if, for any depth-first search of $G$ starting from $s(G)$, the set of cycle edges is invariant.  
\end{thm}

In a flow graph graph $G$, we may write DFS of $G$, as to mean a DFS of $G$ staring from $s(G)$. In addition, if $G$ is  also reducible, based of the above theorem, we may use the terms {\it ancestor} or {descendant} of $G$, as to mean {\it ancestor} or {descendant} of $G-C$, where $C$ is the (unique) set of cycle edges of $G$.  

A {\it topological sort} of a graph $G$ is a sequence $v_1, \ldots, v_n$ of its vertices, such that $v_iv_j \in E(G)$ implies $i < j$. It is well known that $G$ admits a topological sort if and only if $G$ is acyclic. 
Finally, two graphs $G_1, G_2$ are {\it isomorphic} when there is a one-to-one correspondence $f: V(G_1) \cong V(G_2)$ such that $vw \in E(G_1)$ if and only if $f(v)f(w) \in E(G_2)$. In this case,  write $G_1 \cong G_2$, and  call $f$ an {\it isomorphism function} between $G_1,G_2$, with $f(v)$ being the {\it image} of $v$ under $f$.
 
\section{The Graphs of Structured Programming}\label{graphsSP}
In this section, we describe the graphs of structured programming, as established by Dijkstra in~\cite{D72}, leading to the definition  of {\it Dijkstra graphs}. First, we introduce a family of graphs directly related to Dijkstra's concepts of structured programming.

A {\it statement graph} is defined as being one of the following:
\begin{enumerate}[(a)]
	\item {\it trivial} graph
	\item {\it sequence} graph
	\item {\it if} graph
	\item {\it if-then-else} graph
	\item {\it p-case} graph, $p \geq 3$
	\item {\it while} graph
	\item {\it repeat} graph
\end{enumerate}


For our purposes, it is convenient to assign labels to the vertices of statement graphs as follows. Each vertex is either an {\it expansible vertex}, labeled $X$, or a {\it regular vertex}, labelled $R$. See Figures~\ref{statementgraphs1} and~\ref{statementgraphs2}, where the statement graphs are depicted with the corresponding vertex labels. 
All statement graphs are source-sink. Vertex $v$ denotes the source of the graph in each case.

\begin{figure}\label{statementgraphs1}
\centering 
\includegraphics[scale=0.5]{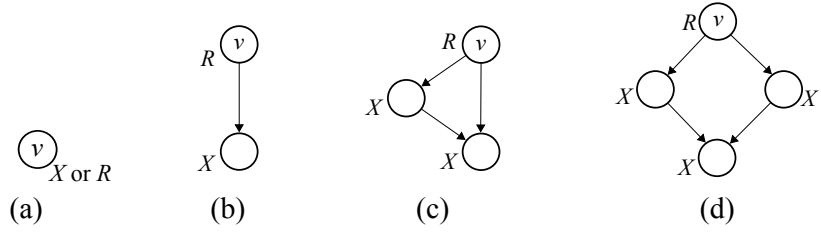}
\caption{Statement graphs (a)-(d)}
\label{statementgraphs1}
\end{figure}

\begin{figure}\label{statementgraphs2}
\centering
\includegraphics[scale=0.5]{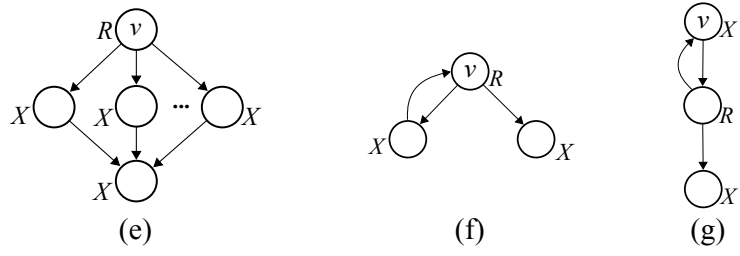}
\caption{Statement graphs (e)-(g)}
\label{statementgraphs2}
\end{figure}

Let $G$ be an unlabeled reducible graph, and $H$ a subgraph of $G$, having source $s(H)$ and sink $t(H)$. We say $H$ is {\it closed} when
\begin{itemize}
	\item $v \in V(H) \setminus s(H)  \Rightarrow N^-(v) \subseteq
	V(H)$;
	\item $v \in V(H) \setminus t(H) \Rightarrow N^+(v)
	\subseteq V(H)$; and 
	\item $v s(H)$ is a cycle edge $\Rightarrow v \in N^+(s(H))$.
\end{itemize}

In this case, $s(H)$ is the only vertex of $H$ having possible  in-neighbors outside $H$, and $t(H)$ the only one possibly having out-neighbors outside $H$. 

The following concepts are central to our purposes.

Let $H$ be an induced subgraph of $G$. We say $H$ is {\it prime} when
\begin{itemize}
	\item $H$ is isomorphic to some non-trivial statement graph, and 
	\item $H$ is closed.
\end{itemize}

It should be noted that the while and repeat graphs, respectively,  (f) and (g) of Figure~\ref{statementgraphs2}, are not isomorphic in the context of flow reducible graphs. In fact, the cycle edge turns them distinguishable. The sources of such graphs are the entry vertices  of the cycle edge, respectively. Then the sink is an out-neighbor of the source in (f), but not in (g).   

Next,  let $G,H$ be two graphs, $V(G) \cap V(H) = \emptyset$, $H$ source-sink, $v \in V(G)$. 

The {\it expansion} of $v$ into a source-sink graph $H$ (Figure~\ref{expansion}) consists of replacing $v$ by $H$, in $G$, such that
\begin{itemize}
	\item $N_G^-(s(H)) := N_G^-(v)$;
	\item $N_G^+(t(H)) := N_G^+(v)$; and
	\item the remaining adjacencies are unchanged.
\end{itemize}

\begin{figure}
\centering
\includegraphics[scale=0.6]{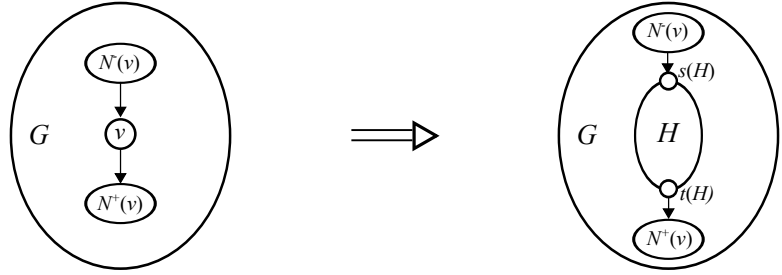}
\caption{Expansion operation}
\label{expansion}
\end{figure}

Now let $G$ be a graph, and  $H$ a prime subgraph of $G$. The {\it contraction} of $H$ into a single vertex (Figure~\ref{contraction}) is the operation defined by the following steps: 
\begin{enumerate}
	\item Identify (coalesce) the vertices of $H$ into the source $s(H)$ of $H$.
	\item Remove all parallel edges and loops.
\end{enumerate}

\begin{figure}
\centering
\includegraphics[scale=0.6]{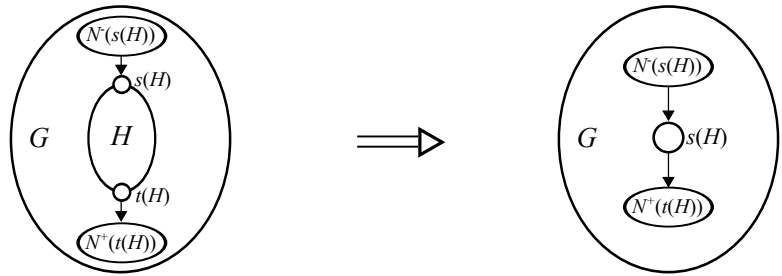}
\caption{Contraction operation}
\label{contraction}
\end{figure}

We finally have the elements to define the class of Dijkstra graphs. The concepts of structured programming and top-down refinement~\cite{D72} lead naturally to the following definition.

A {\it Dijkstra graph (DG)} has vertices labeled $X$ or $R$ recursively defined as:
\begin{enumerate}
	\item A trivial statement graph is a DG.
	\item Any graph obtained from a DG by expanding some $X$-vertex into a non-trivial statement graph is also a DG. Furthermore, after expanding an $X$-labeled vertex $v$ into a statement graph $H$, vertex $s(H)$ is labeled as $R$.  
\end{enumerate}

An example is given in Figure~\ref{DGexample}.

\begin{figure}[h!]
\centering
\includegraphics[scale=0.4]{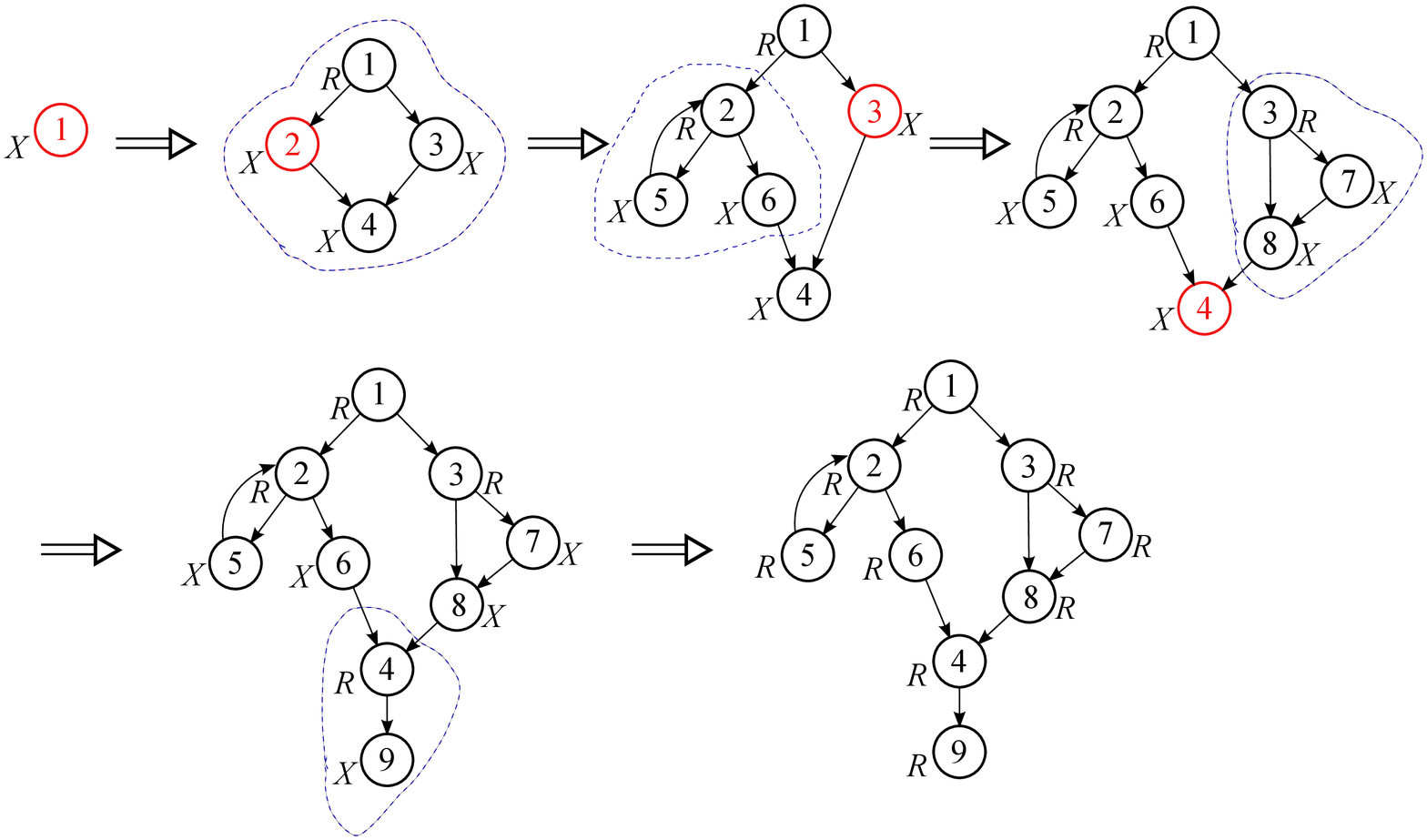}
\caption{Obtaining a Dijkstra graph via vertex expansions}
\label{DGexample}
\end{figure}

The above definition leads directly to a method for constructing  Dijkstra graphs, as follows. Find a sequence of graphs $G_0, \ldots, G_k$, such that 
	\begin{itemize}
		\item $G_0$ is the trivial graph, with the vertex labeled  $X$;
		\item $G_i$ is obtained from $G_{i-1}$, $i \geq 1$, by expanding some X-vertex $v$ of it into a statement graph $H$.
	\end{itemize}



The above construction does not imply a polynomial-time algorithm for recognizing graphs of the class. In the next section, we describe another characterization which leads to such an algorithm. It is relevant to emphasize that the labels are used merely for constructing the graphs. For the actual recognition process, we are interested in the problem of deciding whether a given {\it unlabeled} flow graph is actually a  Dijkstra graph.

\section{Recognition of Dijkstra Graphs}\label{recognition}

In this section, we describe an algorithm for recognizing Dijkstra graphs. For the recognition process, the hypothesis is that we are given an arbitrary flow graph $G$, with no labels, and the aim is to decide whether or not $G$ is a DG. First, we introduce some notation and describe the propositions which form the basis of the algorithm. 

\subsection{Basic Lemmas}
 
 The following lemma states some basic properties of Dijkstra graphs.

 \begin{lem}\label{basic} If $G$ is a Dijkstra graph, then 
	 \begin{enumerate}[(i)]
 		\item  $G$ contains some prime subgraph;
		 \item  $G$ is a source-sink graph; and
		 \item  $G$ is reducible.
	 \end{enumerate}
\end{lem}

\begin{proof}
By definition, there is a sequence of graphs $G_0, \ldots, G_k$, where $G_0$ is trivial, $G_k = G$ and $G_i$ is obtained from $G_{i-1}$ by expanding some $X$-vertex $v_{i-1} \in V(G_{i-1})$ into a statement graph $H_i \subseteq G_i$. Then no vertex $v_i \in V(H_i)$, except $s(H_i)$ has in-neighbors outside $H_i$, and also no vertex $v_i \in V(H_i)$, except $t(H_i)$, has out-neighbors outside $H_i$. 
Furthermore, if $H_i$ contains any cycle then $H_i$ is necessarily a while graph or a repeat graph. The latter implies that such a cycle is $s(H)v$, where $v \in N^+(s(H))$. Therefore $H_i$ is prime in $G_i$ meaning that $(i)$ holds.
To show $(ii)$ and $(iii)$, first observe that any statement graph is single-source and reducible. 
Next, apply induction. 
For $G_0$, there is nothing to prove. 
Assume it holds for $G_i$, $i > 1$. Let $v_{i-1} \in V(G_{i-1})$ be the vertex that expanded into the subgraph $H_i \subseteq G_i$. 
Then the external neighborhoods  of $H_i$ coincide with the neighborhoods  of $v_{i-1}$, respectively.
Consequently, $G_i$ is single-source.
Now, let $C_i$ be any cycle of $G_i$, if existing. If $C_i \cap H_i = \emptyset$ then  $C_i$ is single-entry, since $G_{i-1}$ is reducible. 
Otherwise, if $C_i \subset V(H_i)$ the same is valid, since any statement graph is reducible. 
Finally, if $C_i \not \subset V(H_i)$,  then $v_{i-1}$ is contained in a single-entry cycle $C_{i-1}$ of $G_{i-1}$. Then $C_i$ has been formed by $C_{i-1}$, replacing $v_{i-1}$ by a path contained in $H_i$. 
Since $C_{i-1}$ is single-entry, it follows that $C_i$ must be so.
\end{proof}

Denote by ${\mathcal H}(G)$ the set of non-trivial prime graphs of $G$. 
Let $H,H' \in {\mathcal H}(G)$. Call $H,H'$ {\it independent} when 
\begin{itemize}
	\item $V(H) \cap V(H') = \emptyset$, or
	\item $V(H) \cap V(H') = \{v\}$, where $v = s(H) = t(H')$ or $v = t(H) = s(H')$.
\end{itemize}

The following lemma assures that any pair of distinct, non-trivial prime subgraphs of a graph consists of independent subgraphs.

\begin{lem}\label{independent} Let $H,H' \in {\mathcal H}$. It holds that $H,H'$ are independent.
\end{lem}

\begin{proof}
If $V(H) \cap V(H') = \emptyset$ the lemma holds. Otherwise, let $v \in V(H) \cap V(H')$. The alternatives $v = s(H_1) = s(H_2)$, $v = t(H_1) = t(H_2)$, $v \neq s(H_1),t(H_1)$ or $v \neq s(H_2),t(H_2)$ do not occur because they imply $H_1$ or $H_2$ not to be closed. Next, let $v_1,v_2 \in V(H_1) \cap V(H_2)$, $v_1 \neq v_2$. In this situation, examine the alternative where $v_1 = s(H_1) = t(H_2)$ and $v = s(H_2) = t(H_1)$. The latter implies that exactly one of $H_1$ or $H_2$, say $H_2$, is a while graph or a repeat graph. Then there is a cycle edge $ws(H_1)$, satisfying $w \in N^-(s(H_1))$ and $w \in V(H_2) \setminus \{t(H_2)\}$. Consequently, $w \not \in N^+(s(H_1))$, contradicting $H_1$ to be closed. The only remaining alternative is $V(H_1) \cap V(H_2) = \{v\}$, with $v = s(H_1) = t(H_2)$ or $v = s(H_2) = t(H_1)$. Then $H_1, H_2$ are indeed independent (see Figure~\ref{f:independent}). 
\end{proof}

\begin{figure}
\centering
\includegraphics[scale=0.6]{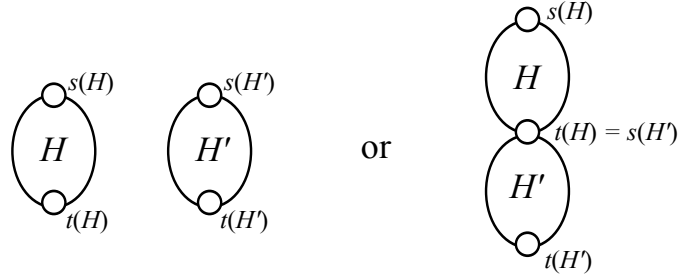}
\caption{Independent primes}
\label{f:independent}
\end{figure}
\vspace{0.5cm}

Next, we introduce a concepts which central for the characterization.

Let $G$ be a graph, ${\mathcal H}(G)$ the set of non-trivial prime subgraphs of $G$, and  $H \in {\mathcal H}(G)$.
Denote by $G \downarrow H $ the graph obtained from $G$ by contracting $H$. For $v \in V(G)$, the {\it image} of $v$ in $G \downarrow H$, denoted $I_{G \downarrow H}(v)$, is

\vspace{-0.3cm}

\[I_{G \downarrow H}(v) = \left\{\begin{array}{ll}
    v, &\mbox{if } v \not \in V(H) \\
    s(H), &\mbox{otherwise.}
    \end{array} \right. \]

For $V' \subseteq V(G)$, define  the ({\it subset}) {\it image} of $V'$ in $G \downarrow H$, as $I_{G \downarrow H}(V') = \cup_{v \in V'} I_{G \downarrow H}(v)$. Similarly,  for $H' \subseteq G$, the ({\it subgraph})
{\it image} of $H'$ in $G \downarrow H$, denoted by
$I_{G \downarrow H}(H')$, is the  subgraph induced in $G \downarrow H$
by the subset of vertices $I_{G \downarrow H}(V(H'))$.

The following lemmas  are employed in the ensuing characterization. The first shows that any prime subgraph $H \in {\mathcal G}$ is preserved under contractions of different primes. 
Let $G$ be an arbitrary flow graph, 
$H, H' \in {\mathcal H}(G)$, $H \neq H'$.

\begin{lem}\label{prime-preservation}
$I_{G \downarrow H}(H') \in {\mathcal H}(G \downarrow H)$.
\end{lem}

\begin{proof}
Let $G$ be a graph,  
$H, H' \in {\mathcal H}(G)$, $H \neq H'$. By Lemma~\ref{independent}, $H,H'$ are independent. If $H,H'$ are disjoint the contraction of $H$ does not affect $H'$, and the lemma holds. Otherwise, by the independence condition, it follows that $V(H) \cap V(H') = \{v\}$, where $v = s(H) = t(H')$ or $v = s(H') = t(H)$. Examine the first of these alternatives. By contracting $H$, all neighborhoods of the vertices of $I_{G \downarrow H}(H')$ remain unchanged, except that of $I_{G \downarrow H}(s(H'))$, since its in-neighborhood becomes equal to $N_G^-(s(H))$. On the other hand, the contraction of $H$ into $v$ cannot introduce new cycles in $H'$.
Consequently, $H'$ preserves in $G \downarrow H$ its property of being a non-trivial and closed statement graph, moreover, prime. Finally, suppose $v = s(H)=t(H')$. Again, the neighborhoods of the vertices of $I_{G \downarrow
H}(H')$ are preserved, except possibly the out-neighborhoods of the vertices of $I_{G \downarrow H}(t(H'))$, which become $N_G^+(t(H))$, after possibly removing self-loops. Consequently,   $I_{G \downarrow H}(H')  \in {\mathcal H}(G \downarrow H)$. 
\end{proof}

Next we prove  prove a commutative law for the order of contractions. 

\begin{lem}\label{Commutative-law}
If $H, H' \in {\mathcal H}(G)$, then
$(G \downarrow H) \downarrow (I_{G \downarrow H}(H')) \cong
  (G \downarrow H') \downarrow (I_{G \downarrow H'}(H)).$
\end{lem}

\begin{proof}
Let $A \cong (G \downarrow H) \downarrow (I_{G \downarrow H}(H'))$ and $B \cong (G \downarrow H') \downarrow (I_{G \downarrow H'}(H))$. 
By Lemma~\ref{independent}, $H,H'$ are independent. 
First, suppose $H,H'$ are disjoint. 
Then $I_{G \downarrow H}(H') = H'$ and $I_{G \downarrow H'}(H) = H$. 
It follows that, in both graphs $A$ and $B$, the subgraphs $H$ and $H'$ are respectively replaced by a pair of non-adjacent vertices, whose in-neighborhoods are $N_G^-(s(H))$ and $N_G^-(s(H'))$, and out-neighborhoods $N_G^+(t(H))$ and $N_G^+(t(H'))$, respectively. Then $A = B$. 
In the second alternatives, suppose $H,H'$ are not disjoint. 
Then $V(H) \cap V(H') = \{v\}$, where $v = s(H) = t(H')$, or $v = t(H) = s(H')$. In both cases, and in both graphs $A$ and $B$, the subgraphs $H$ and $H'$ are contracted into a common vertex $w$. 
When $v = s(H) = t(H')$, it follows $N_G^-(A) = N_G^-(s(H')) = N_B^-(v)$ and $N_A^+(v) = N_G^+(t(H)) = N_B^+(v)$. Finally, when $v = t(H) = s(H')$, we have $N_A^-(v) = N_G^-(s(H)) = N_B^-(v)$, while $N_A^+(v) = N_G^+(t(H')) = N_B^+(v)$. Consequently, $A = B$ in any situation. $\square$
\end{proof}

\subsection{Contractile Sequences}



A sequence of graphs $G_0, \ldots, G_k$ is a {\it contractile
sequence} for a graph $G$, when 
\begin{itemize}
	\item $G \cong G_0$, and
	\item $G_{i+1} \cong (G_i \downarrow H_i)$, for some $H_i \in {\mathcal H}(G_i)$, $i < k$. Call $H_i$ the {\it contracting prime} of $G_i$.
\end{itemize}

We say $G_0, \ldots, G_k$ is {\it maximal} when ${\mathcal H}(G_k) = \emptyset$.
In particular, if $G_k$ is the trivial graph then $G_0, \ldots,
G_k$ is maximal.\\

Let $G_0, \ldots, G_k$, be a  contractile sequence of $G$, and
$H_j$ the contracting prime of $G_j$. That is,  
$G_{j+1} \cong (G_j \downarrow H_j$), $0 \leq j < k$.   
For $H_j' \subseteq G_j$ and $q \geq j$, the {\it iterated image} of $H_j'$ in $G_q$ is recursively defined as 
\vspace{-0.3cm}

\[I_{G_q}(H'_j) = \left\{\begin{array}{ll}
  H'_j,  &\mbox{if } q = j \\
   I_{G_q}(I_{G_{j+1}}(H'_j)), &\mbox{otherwise.}
    \end{array} \right. \]

Finally, we describe the characterization in which the recognition algorithm for Dijkstra graphs is based.

\begin{thm}\label{Main-characterization}
Let $G$ be an arbitrary flow graph, with $G_0, \ldots, G_k$ and $G'_0,
\ldots, G'_{k'}$ two contractile sequences of $G$. Then $G_k \cong G'_{k'}$. Furthermore, $k = k'$.
\end{thm}

\begin{proof}
Let $G_0, \ldots, G_k$ and $G'_0, \ldots, G'_{k'}$ be two contractile sequences, denoted respectively by $S$ and $S'$ of a graph $G$. Let $H_j$ and $H'_{j}$ be the contracting primes of $G_j$ and $G'_j$, respectively. That is, $G_{j+1} \cong (G_j \downarrow H_j)$ and $G'_{j+1} \cong (G'_j \downarrow H'_j)$, $j < k$ and $j < k'$. Without loss of generality, assume $k \leq k'$. Let $i$ be the least index, such that $G_j \cong G'_j$, $j \leq i$. Such an index exists since $G \cong G_0 \cong G'_0$. If $i = k$ then $G_k \cong G'_{k'}$, implying $k = k'$ and the theorem holds. 
Otherwise, $i < k$, $G_i \cong G'_i$ and $G_i \not \cong G'_i$. Since $G_i \cong G'_i$, it follows $H_i \in {\mathcal H} (G'_i)$. By Lemma~\ref{prime-preservation}, the iterated image $H_{i_{q}}$, of $H_i$ in $G'_q$ is preserved as a prime subgraph for all $G'_q$, as long as it does not become the contracting prime of $G'_{q-1}$. Since $G'_{k'}$ has no prime subgraph, it follows there exists some index $p$, $i < p < k'$, such that $G'_{p+1} \cong (G_p \downarrow H_{i_{p}})$, where $H_{i_{p}}$ represents the iterated image of $H_i$ in $G'_p$. Let $H_{i_{p-1}}$ be the iterated image of $H_i$ in $G'_{p-1}$. Clearly, $H'_{p-1}, H_{i_{p-1}} \in {\mathcal H}(G'_{p-1})$, and by Lemma~\ref{independent}, $H'_{p-1}$ and $H_{i_{p-1}}$ are independent in $G'_{p-1}$. Since $((G'_{p-1} \downarrow H'_{p-1}) \downarrow H_{i_{p}}) \cong G'_{p+1}$, by Lemma~\ref{prime-preservation}, it follows that $((G'_{p-1} \downarrow H_{i_{p-1}}) \downarrow H''_{p-1}) \cong G'_{p+1}$, where $H''_{p-1}$ represents the image of $H'_{p-1}$ in $G'_{p-1} \downarrow H_{i_{p-1}}$. Consequently, we have exchanged the positions in $S'$ of two contracting primes, respectively at indices $p-1$ and $p$, while preserving all graphs $G'_q$, for $q < p-1$ and $q > p$. In  particular, preserving the graph $G'_{p+1}$ and all graphs lying after $G'_{p+1}$ in $S'$, together with their corresponding contracting primes.

Finally, apply the above operation iteratively, until eventually the iterated image of $H_i$ becomes the contracting prime of $G'_i$. In the latter situation, the two sequences  coincide up to
index $i+1$, while preserving the original graphs $G_k$ and $G'_{k'}$. Again, applying iteratively such an argument, we eventually obtain that the two sequences turned coincident, preserving the original graphs $G_k$ and $G'_{k'}$. Consequently, $G_k \cong G'_{k'}$  and $k = k'$.
\end{proof}

\subsection{The Recognition Algorithm}\label{algorithm}

We start with a bound for the number $m$ of edges of Dijkstra graphs.

\begin{lem}\label{bound}  Let $G$ be a DG graph. Then $m \leq 2n -2$.
\end{lem}

\begin{proof}:
If $G$ is a DG graph there is a sequence of graphs $G_0, \ldots G_k$, where $G_0$ is the trivial graph, $G_k \cong G$ and $G_i$ is obtained from $G_{i-1}$ by expanding an $X$-vertex of $G_{i-1}$ into a statement graph. Apply induction on the number of expansions employed in the construction of $G$. If $k = 0$ then $G$ is a trivial graph, which satisfies the lemma. For $k \geq 0$, Suppose the lemma true for any graph $G' \cong G_i$, $i < k$.  In particular, let $G_i \cong G_{k-1}$. Let $n'$ and  $m'$ be the number of vertices and edges of $G'$, respectively. Then $m' \leq 2n' -2$. We know that $G_k$ has been obtained by expanding a vertex of $G_{k-1}$ into a statement graph $H$. Discuss the alternatives for $H$. If $H$ is the trivial graph then $n = n'$ and $m = m'$.  If $H$ is a sequence graph then $n = n' + 1$ and $m = m' + 1$. If $H$ is an if graph, a while graph or repeat graph then $n = n' + 2$ and $m = m' + 3$. If $H$ is an if then else graph or a $p$-case graph then $n = n' + p + 1$ and $m = m' + 2p$, where $p$ is the outdegree of the source of $H$. In any of these alternatives, a simple calculation implies $m \leq 2n-2$.
\end{proof}

We can describe an algorithm for recognizing Dijkstra graphs based on Theorem~\ref{Main-characterization}. We recall that the input is a unlabeled flow graph with no labels. Furthermore, for a while, assume that $G$ is reducible, otherwise by Lemma~\ref{basic} it is surely not a Dijkstra graph.

Let $G$ be a flow reducible graph. To apply Theorem~\ref{Main-characterization}, we construct a contractile sequence $G_0, \ldots, G_k$ of $G$. That is, find iteratively a non-trivial prime subgraph $H_i$ of the $G_i$ and  contract it, until either the graph becomes trivial or otherwise no such subgraph exists. In the first case the graph is a DG, while in the second it is not. Recall from Lemma~\ref{prime-preservation} that whenever $G_i$  contains another prime $H_j \neq H_i$ then the iterated image of $H_j$ is preserved, as long as it does not become the contracting prime in some later iteration. On the other hand, the contraction $G_i \downarrow H_i$ may generate a new prime $H'_i$, as shown in Figure~\ref{fig:generator}. However, the generation of new primes obeys a rule, described by the lemma below.

\begin{figure}[h]  
    \centering \includegraphics[height=5cm]{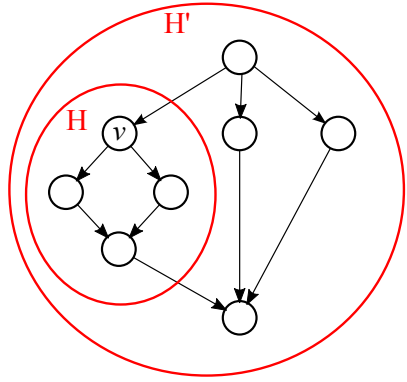}
     \label{fig:generator}
		  \caption{Generating a new prime H'}
\end{figure}   

\begin{lem}\label{lem:generator-generated}
Let $G$ be  reducible graph, $H \in {\mathcal H}(G)$,  $H' \in {\mathcal H}(G \downarrow H)
\setminus {\mathcal H}(G)$. Then $s(H)$ is a proper descendant of $s(H')$ in $G \downarrow H$.
\end{lem}

The above lemma suggests us to consider special contractile sequences, as below. 

Let $G$ be a reducible graph,  $G_0, \ldots, G_k$ a  contractile
sequence ${\mathcal C}$ of $G$,  $H_i$ the contracting prime
of $G_i$, $0 \leq i < k$. Say that ${\mathcal C}$ is a {\it bottom-up
(contractile) sequence} of $G$ when 
each contracting prime  $H_i$ satisfies: 
$s(H_i)$ is not a descendant of $s(H)$, for any prime $H \neq H_i$
of $G_i$.

The idea of the recognition algorithm then becomes as follows. Let $G$ be a reducible graph. Iteratively, find a lowest
vertex $v$ of $G$, s.t. $v$ is the source of a prime
subgraph $H$ of $G$. Then contract $H$. Stop when noprimes exist any more. 

A complete description of the algorithm is below detailed. The algorithm answers YES or NO, according to respectively $G$ is a Dijkstra graph or not.

\begin{algorithm}\caption{Dijkstra graphs recognition algorithm}\label{alg:recognition}
$G$, arbitrary flow graph (no labels) \\
Count the number $m$ of edges of $G$. If $m \geq 2n-1$  then {\bf return} NO \\
$E_C$, set of cycle edges of a DFS of $G$, starting at $s(G)$ \\
$v_1, \ldots, v_n$,  topological sorting of $G - E_C$ \\
$i := n$ \\
{\bf while} $i \geq 1$ {\bf do} \\
\hspace*{.5cm} {\bf if} $G$ is the trivial graph  \\
\hspace*{1cm} {\bf then} {\bf return} YES,  {\bf stop} \\
\hspace*{.5cm} if $v_i$ is the source of a prime subgraph $H$ of
$G$ \\
\hspace*{1cm} {\bf then} $G := G \downarrow H$ \\
\hspace*{.5cm} $i : = i-1$ \\
{\bf return} NO
\end{algorithm}

The correctness of Algorithm~\ref{alg:recognition} follows basically from Theorem~\ref{Main-characterization} and Lemma~\ref{lem:generator-generated}. However, the latter relies on the fact that $G$ is a reducible graph, whereas the proposed algorithm considers as input an arbitrary graph. The lemma below justifies that can we avoid the step of recognizing reducible graphs.

\begin{lem}\label{lem:avoid-rec-reducible}
Let $G$ be an arbitrary flow graph input to Algorithm~\ref{alg:recognition}. If $G$ is not a reducible graph then the algorithm would correctly answer NO.
\end{lem}

\begin{proof}
If $G$ is not a reducible graph let $E_C$ be the set of cycle edges, relative to some DFS startingate $s(G)$. Then $G$ contains some cycle $C$, such that $w$ does not separate $s(G)$ from $v$, where $vw \in E_C$ is the cycle edge of $C$. Without loss of generality,  consider the inner most of these cycles. The only way in which the edge $vw$, or any of its possible images, can be contracted is in context the  of a while or repeat prime subgraph $H$, in which the cycle would be contracted into vertex $w$, or a possible iterated image of it.  However there is no possibility for $H$ to be identified as such, because the edge entering the cycle from outside prevents the subgraph  to be closed. Consequently, the algorithm necessarily would answer NO.
\end{proof}

As for the complexity, first observe that to decide whether the graph contains a non-trivial prime subgraph whose source is a given vertex $v \in V(G)$, we need  $O|(N^+(v)|$ steps. Therefore, when considering all vertices of $G$ we require $O(m)$ time. There can be $O(n)$ prime subgraphs altogether, and each time some prime $H$ is identified, it is contracted, and the size of the graph decreases by $|E(H)|$. The number of steps required to contract a $H$ is $O|E(H)|$. Hence each edge is examined at most a constant number of times during the entire process.  Finding a topological sorting of a graph can be done in $O(m)$.  Thus, the time complexity is $O(m)$, that is, $O(n)$, by Lemma~\ref{bound}.

\section{Isomorphism of Dijkstra Graphs}\label{sec:isomorphism}

In this section, we describe a linear time algorithm for the isomorphism of Dijkstra graphs. 

Given a Dijkstra graph $G$, the general idea consists of defining a code $C(G)$ for $G$, having the following property. For any two Dijkstra graphs $G_1,G_2$, $G_1 \cong G_2$ if and only if $C(G_1) = C(G_2)$.  

As in the recognition algorithm, the codes are obtained by constructing  a bottom-up contractile sequence of each graph.     
The codes refer explicitly to the statement graphs having source $v$ as depicted in Figures~\ref{statementgraphs1} and~\ref{statementgraphs2}, and consist of (linear) strings. 
For a Dijkstra graph $G$, the string $C(G)$ that will be coding $G$ is constructed over an alphabet of symbols containing integers in the range  $\{1, \ldots, \Delta^+(G)+4\}$, where $\Delta^+(G)$ is the maximum cardinality among the out-neighborhoods of $G$. Let, $A,B$ be a pair of strings. The concatenation of $A$ and $B$, denoted $A || B$, is the string formed by $A$, immediately followed by $B$.

In order to define the code $C(G)$ for a Dijkstra graph $G$, we assign an integer, named $type(H)$, for each statement graph $H$, a code $C(v)$ for each vertex $v \in V(G)$, and a code $C(H)$ for each prime subgraph $H$ of a bottom-up contractile sequence of $G$. The code $C(G)$ of the graph $G$ is defined as being that of the source of $G$. For a subset $V' \subseteq V(G)$, the code $C(V')$ of $V'$ is the set of strings $C(V') = \{C(v_i) | v_i \in V'\}$. Write $lex(C(V')) = C(v_1) || ... || C(v_r)$ whenever $V' = \{v_1, \ldots, v_r\}$ and $C(v_i)$ is lexicographically not greater than $C(v_{i+1})$.   

\normalsize
\begin{table}[h!]\caption{Statement graph types and codes $C(H)$ of prime subgraphs $H$}\label{t:types-primes}
\centering
\begin{tabular}{|c|c|l|}
\hline {\bf statement} &   {{\bf $type(H)$}} & {\bf $C(H), v=s(H)$}\\
{\bf graphs} $H$ & & \\
 \hline 
trivial                     & 1         &
\\ \hline sequence              & 2         & $2 || C(N^+(v))$
\\ \hline if-then                   & 3          & $3 || C(N^+(v))
\setminus N^{+2}(v))|| C(N^{+2}(v))$            \\ \hline while &
4 & $4 || C(N^+(v) \cap N^-(v)) || C(N^+(v) \setminus N^-(v))$
\\ \hline repeat                   & 5          & $5
|| C(N^+(v)) || C(N^{+2}(v) \setminus \{v\})$      \\
\hline if-then-else            & 6          & $6 || lex(C(N^+(v)))
|| C(N^{+2}(v))$                    \\ \hline $p$-case & $p+4$  &
$p+4 || lex(C(N^+(v))) || C(N^{+2}(v))$
\\ \hline
\end{tabular}
\end{table}
\normalsize

Next, we describe how to obtain the actual codes. The types of the the different statement graphs are shown in the second column of Table~\ref{t:types-primes}. For a vertex $v \in V(G)$, the code $C(v)$ is initially set to 1. Subsequently, if $v$ becomes the source of a prime graph $H$, the string $C(v)$ is updated by implicitly assigning 
$C(v) := C(v) || C(H),$
where  $C(H)$ is given by the third column of the table.
Such an operation is called the {\it expansion} of $v$. 
It follows that $C(H)$ is written in terms of $type(H)$ and the codes of the vertices of $H$, and so on iteratively. 
A possible expansion of some other vertex $w \in V(G)$ could imply in an expansion of $v$, and so iteratively.  
Observe that when $H$ is an if-then-else or a $p$-case graph, we have chosen to place the codes of the out-neighbors of $s(H)$ in lexicographic ordering. For the remaining statement graphs $H$, the ordering of the codes of the out-neighbors of $s(H)$ is also unique and implicitly imposed by $H$. When all primes associated to $C(v)$ have been expanded, $C(v)$ has reached its final value,

\subsection{The Isomorphism Algorithm}

Next, we describe the actual formulation of the algortithm. 

Let $G$ be a DG. Algorithm~\ref{alg:isomorphism} constructs the encoding $C(G)$ for $G$.

\begin{algorithm}[h!]\caption{Dijkstra graphs isomorphism algorithm}\label{alg:isomorphism}
\smallskip
\small $G$, DG; $E_C$, set of cycle edges of $G$\\
Find a topological sorting $v_1, \ldots, v_n$ of $G-E_C$ \\
{\bf for} $i = n, n-1, \ldots, 1$ {\bf do} \\
\hspace*{.3cm} $C(v_i) := 1$ \\
\hspace*{.3cm} {\bf if} $v_i$ is the source of a prime subgraph $H$ {\bf then} \\
 \hspace*{.7cm} $C(v_i) := C(v_i) || \left\{
\begin{array}{ll}
\displaystyle  2 || C(N^+(v_i)), ~\textrm{if}~H~\textrm{is a sequence graph;} \\
\displaystyle  3 || C(N^+(v_i) \setminus N^{+2}(v_i))||
C(N^{+2}(v_i)), \\
     \hspace{2.9cm}  ~\textrm{if}~H~\textrm{is an if-then graph;} \\
\displaystyle  4 || C(N^+(v_i) \cap N^-(v_i)) || C(N^+(v_i) \setminus N^-(v_i)), \\
    \hspace{2.9cm}  \textrm{if}~H~\textrm{is a while graph,} \\
\displaystyle  5 || C(N^+(v_i)) || C(N^{+2}(v_i) \setminus \{v_i\})
, \\
    \hspace{2.9cm}  \textrm{if}~H~\textrm{is a repeat graph;} \\
\displaystyle  6 || lex(C(N^+(v_i))) || C(N^{+2}(v_i)), \\
    \hspace{2.9cm} \textrm{if}~H~\textrm{is an if-then-else
    graph.} \\
\displaystyle p+4 || lex(C(N^+(v_i))) || C(N^{+2}(v_i)), \\
    \hspace{2.9cm} \textrm{if}~H~\textrm{is a p-case graph.}
\end{array}
\right.
$
$C(G):=C(v_1)$
\end{algorithm}

\normalsize




An example is given in Figure~\ref{f:SPexalg}.

\begin{figure}[h!]
\centering
\includegraphics[scale=0.55]{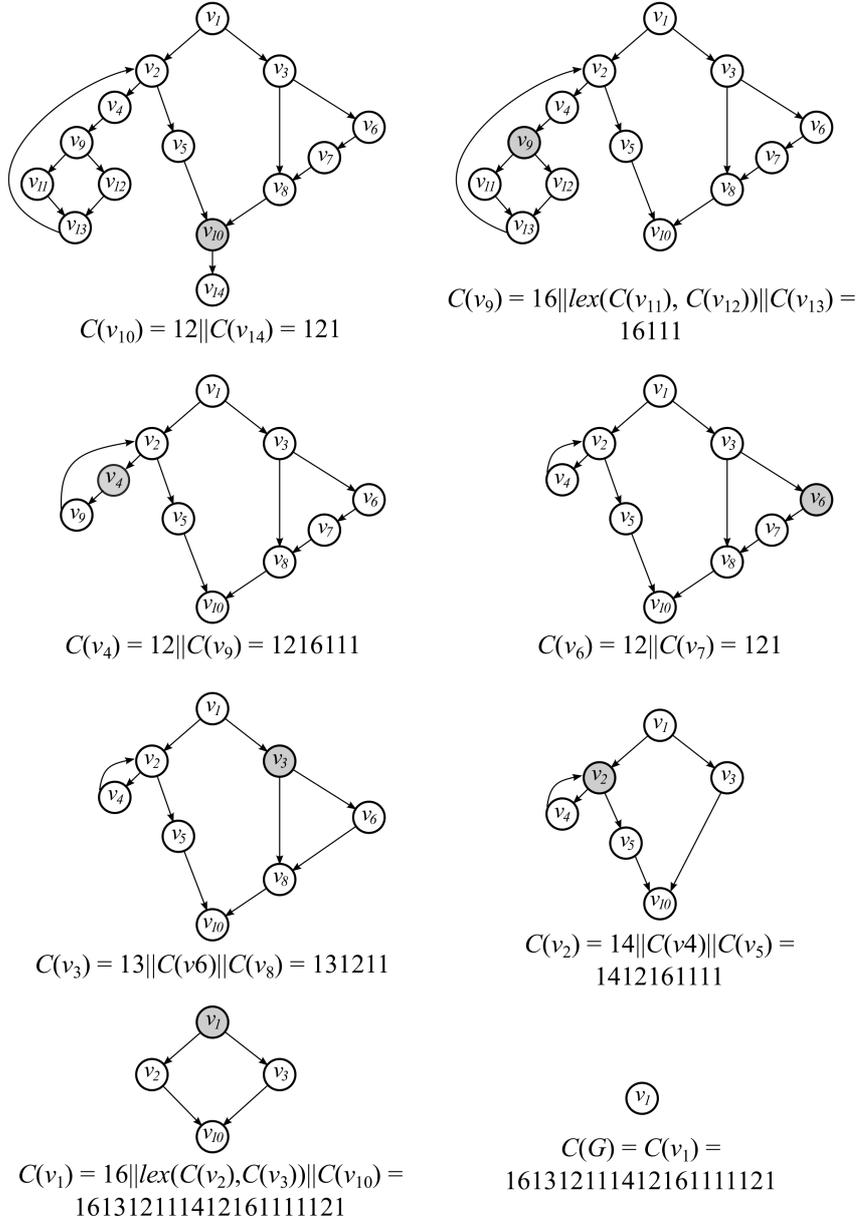}
\caption{Example for isomorphism algorithm}
\label{f:SPexalg}
\end{figure}





\subsection{Correctness and Complexity}

\begin{thm}\label{thm:isomorphism}
Let $G, G'$ de Dijkstra graphs, and $C(G), C(G')$ their codes, respectively. Then $G,G'$ are isomorphic if and only if $C(G) = C(G')$.
\end{thm}

\begin{proof}
By hypothesis,  $G, G'$ are isomorphic. We show that it implies $C(G) = C(G')$.  Following the isomorphism algorithm, observe that the number of 1's in the strings $C(G), C(G')$ represents the number of vertices of $G,G'$, respectively, 
whereas each integer $> 1$ in the strings, represents the contraction of a prime subgraph. Furthermore, each prime subgraph $H$, which is initially contained in the input graph $G$, corresponds in $C(G)$, to a substring formed by the integer $type(H)$ followed by one 1, if $type(H) = 2$; or two 1's, if $type(H)=3$; or three 1's, if $ 4 \leq type(H) \leq 6$; or $type(H)+1$ 1's, if $type(H) > 6$; respectively. Clearly, the same holds for the graph $G'$ and its code $C(G')$.  The proof is by induction on the number $k$ of contractions needed to reduce both $G$ and $G'$ to a trivial vertex. By Theorem~\ref{Main-characterization}, $k$ is invariant and applies for both graphs $G$ and $G'$. If $k = 0$ then both $G$ and $G'$ are trivial graphs, and the  theorem holds, since $C(G) = C(G') = 1$. When $k > 0$, assume that if $G_-$ and $G_-'$ are isomorphic  DG graphs which require less than $k$ contractions for reduction then $C(G_-) = C(G_-')$. Furthermore, assume also by the induction hypothesis, that if $v,v'$ are vertices of $G_-,G_-'$,  corresponding to 1's at the same relative positions in $C(G)$ and $C(G_-)$, respectively, then $v' = f(v)$, where $f$ is the isomorphism function between $G_-$ and $G_-'$. Now, consider the graphs $G$ and $G'$. Choose a prime subgraph $H$ of $G$, and let $v = s(H)$. Let $v' = f(v)$ be a vertex of $G'$ corresponding to $v$ by the isomorphism. Since $G \cong G'$, it follows that $v'$ is the source of a prime subgraph $H'$ of $G'$. Moreover $H \cong H'$. Consider the contractions $G \downarrow H$ and $G' \downarrow H'$, leading to graphs $G_-$ and $G_-'$, respectively. Let $C_-(G)$ and $C_-(G')$ be the strings obtained from $C(G)$ and $C(G')$, respectively by contracting the substrings corresponding to $H$ and $H'$, as above. That is, all the 1's of $C(H)$ and $C(H')$ are compressed into the positions of $v=s(H)$ and  $v'=s(H')$, respectively, while the integers $type(H)$ and $type(H')$ become 1, maitaining their original positions.  It follows that $C(G_-) = C_-(G)$ and $C(G_-') = C_-(G')$. By the induction hypothesis $C(G_-) = C(G_-')$ and the 1's corresponding to $v$ and $v'$ lie in the same relative positions in the strings. Consequently, by replacing the latter 1's for the substrings which originally represented $H$ and $H'$, we conclude that indeed $C(G) = C(G')$, and moreover the induction hypothesis is still verified. The converse is similar.
\end{proof}

The corollaries below are direct consequences of Theorem~\ref{thm:isomorphism}.

\begin{cor}\label{cor:unique}
Let $G$ be a DG. The following affirmatives hold. 
\begin{enumerate}
\item There is a one-to-one correspondence between the 1's of $C(G)$ and the vertices of $G$. 
\item The code $C(G)$ of G  is unique and is a representation of $G$.
\end{enumerate}
\end{cor}

\begin{cor}\label{cor:1-1}
Let $G,G'$ be DGs and $C(G),C(G')$ their corresponding codes, satisfying $C(G) = C(G')$. Then an isomorphism function $f$ between $G$ and $G'$ can be determined as follows. Let $v \in V(G)$ and $v' \in V(G')$ correspond to 1's at identical relative positions in $C(G)$ and $C(G')$, respectively. Define $f(v) := v'$. 
\end{cor}

Finally, consider the complexity of the isomorphism algorithm. 

\begin{lem}\label{lem:length}
Let $G$ be a Dijkstra graph, and $C(G)$ its code. 
Then $|C(G)| = n + k \leq 2n -1$, where $n$ is the number of vertices of $G$ and $k$ the number of contractions needed to reduce it to a trivial vertex. 
\end{lem}

\begin{proof}
The encoding $C(G)$ consists of exactly $n$ 1's, together with elements of a multiset $U \subseteq \{2,3, \ldots, \Delta^+(G) + 4\}$. We know that $C(G)$ starts and ends with an 1, and it contains no two consecutive elements of $U$. Therefore $C(G) \leq 2n -1$. When $G$ consists of the induced path $P_n$, it follows $|C(P_n)| = 2n-1$, attaining the bound. 
\end{proof}

\begin{thm}\label{thm:complexity-iso} The isomorphism algorithm terminates within $O(n)$ time.
\end{thm}

\begin{proof}
Recall that $m = O(n)$, by Lemma~\ref{bound}.  The construction of a bottom-up contractile sequence requires 
 $O(n)$ steps.  For each $v \in V(G)$, following the isomorphism algorithm, $C(v)$ can be constructed in time $|C(v)|$. We remark that 
lexicographic ordering takes linear time on the total length of the strings to be sorted. It follows that the algorithm requires no more than $O(n)$ time to construct the code $C(G)$ of $G$. 
\end{proof}

\section{Conclusions}\label{conclusions}

The analysis of control flow graphs and different forms of structuring  have been considered in various papers. 
To our knowledge, no full characterization and no recognition algorithm for control flow graphs of structured programs have been described before.  There are some related classes for which characterizations and efficient recognition algorithms do exist, e.g.~the classes of reducible graphs and D-charts. However, both contain and  are much larger than Dijkstra graphs.   

An important question solved in this paper is that of recognizing whether two control flow graphs (of structured programs) are syntactically equivalent, i.e., isomorphic. Such question fits in the area of  \emph{code similarity analysis}, with applications in clone detection, plagiarism and software forensics.

Since the establishment of structured programming, some new statements have been proposed to add to the original structures which forms the classical structured programming, enlarging the collection of allowed statements. Some of such relevant statements are depicted in Figures~\ref{fig:gDG}.
\begin{enumerate}[(a)]
	\item {\it break-while}: 
	Allows an early exit from a while statement;
	\item {\it continue-while}:
	Allows a while statement to proceed, after its original termination;
	\item {\it break-repeat}: 
	Allows an early exit from a repeat statement;
	\item {\it continue-repeat}: 
	Allows a repeat statement to proceed, after its original termination;
	\item {\it divergent-if-then-else}: 
	A selection statement, similar to the standard {\it if-then-else}, except that the \-comparisons do not converge 
afterwords to a same point, but lead to disjoint structures. Note that the corresponding graph has no longer a (unique) sink. 
\end{enumerate}

\begin{figure}
    \centering \includegraphics[scale=0.35]{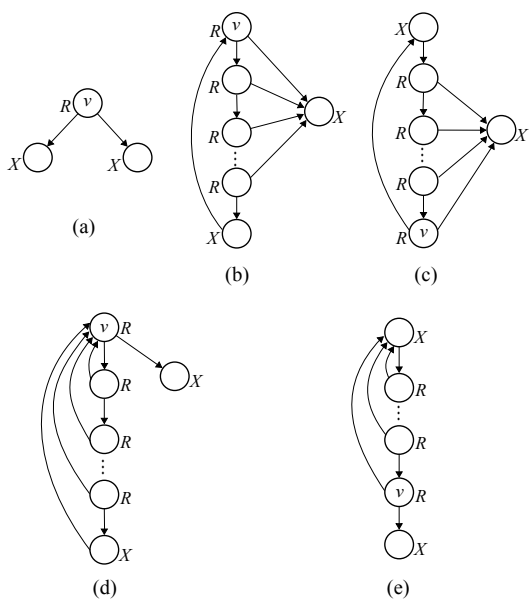}
    \caption{Generalized Dijkstra graphs}\label{fig:gDG}
\end{figure}



In fact, the inclusion of some of the above additional control blocks in structured programming has been already predicted in some papers, as~\cite{K74}. The basic ideas and techniques described in the present work can be generalized, so as to efficiently recognize graphs that incorporate the above statements, in addition to those of Dijkstra graphs. Similarly, for the isomorphism algorithm.

\vspace{2cm}

\small
\centerline{Acknowledgments} 
The authors are grateful to Victor Campos for the helpful discussions and comments during the French-Brazilian Workshop of Graphs and Optimizations, in Redonda, CE, Brazil, 2016. He pointed out the possibility of decreasing the complexity of the recognition algorithm from $O(n^2)$ to $O(n)$. 
\normalsize






\end{document}